\documentclass[copyright]{eptcs}
\providecommand{\event}{DCFS 2010} 
\usepackage{breakurl}             
\usepackage{vaucanson-g}
\usepackage{amsthm,amsmath,amsfonts}

\DeclareMathOperator{\rep}{rep}
\DeclareMathOperator{\val}{val}

\DeclareMathOperator{\N}{\mathbb N}

\newtheorem{theorem}{Theorem}
\newtheorem{lemma}[theorem]{Lemma}
\newtheorem{corollary}[theorem]{Corollary}
\newtheorem{proposition}[theorem]{Proposition}
\theoremstyle{definition}
\newtheorem{definition}[theorem]{Definition}
\newtheorem{remark}[theorem]{Remark}
\newtheorem{example}[theorem]{Example}

\newcommand{\bdisplay}{\begin{description}\footnotesize\item[]}
\newcommand{\edisplay}{\end{description}}

\newcommand{\bquot}[1]{\begin{quotation}\small\noindent
  \textbf{#1}\hspace{\labelsep}\ignorespaces}
\newcommand{\equot}{\unskip\end{quotation}}

\title{State Complexity of Testing Divisibility}
\author{Emilie Charlier, Narad Rampersad, Michel Rigo, Laurent Waxweiler
\institute{Department of Mathematics, University of Li{\`e}ge\\
Grande Traverse 12 (B37), B-4000 Li{\`e}ge, Belgium}
\email{echarlier@ulg.ac.be, nrampersad@ulg.ac.be,
M.Rigo@ulg.ac.be, L.Waxweiler@ulg.ac.be}
}
\def\titlerunning{State Complexity of Testing Divisibility}
\def\authorrunning{E.~Charlier, N.~Rampersad, M.~Rigo, L.~Waxweiler}
\begin{document}
\maketitle

\begin{abstract}
    Under some mild assumptions, we study the state complexity of the
    trim minimal automaton accepting the greedy representations of the
    multiples of $m\ge 2$ for a wide class of linear numeration
    systems. As an example, the number of states of the trim minimal
    automaton accepting the greedy representations of $m\N$ in the
    Fibonacci system is exactly $2m^2$\@.

\end{abstract}

\section{Introduction}\label{s:in}

Cobham \cite{Cobham} showed that ultimately periodic sets of non-negative integers are the only sets that are recognized by a finite automaton in every integer base numeration system. The
ultimately periodic sets are also exactly the sets definable in the
Presburger arithmetic $\langle\mathbb{N},+\rangle$. In the context of a non-standard numeration system $U$, if $\mathbb{N}$ is $U$-recognizable, then $U$ is easily seen to be a
linear numeration system \cite{Shallit}. For linear numeration systems, ultimately periodic sets are all recognized by finite automata if and only if $\N$ is (see Theorem~\ref{the:folk} below). Conditions on a linear numeration system $U$ for $\mathbb{N}$ to be $U$-recognizable are considered in \cite{Hollander}. Among linear numeration systems for which $\mathbb{N}$ is $U$-recognizable, the class of systems whose characteristic polynomial
is the minimal polynomial of a Pisot number has been widely studied
\cite{BH}. An example of such a system is given by the 
Fibonacci numeration system (see Example~\ref{exa:1}). 

Let $U$ be a linear numeration system and $X$ be a
$U$-recognizable set of non-negative integers given by some DFA recognizing
the greedy representations of elements of $X$. For integer base systems, Honkala has proved that one
can decide whether or not $X$ is ultimately periodic \cite{Hon}. Another,
shorter proof of this result can be found in \cite{ARS}. 
For a wide class of linear numeration systems, the same decidability
question is answered positively in \cite{CR,BCFR}.  For all the above
mentioned reasons ultimately periodic sets of integers and, in
particular, the recognizability of a given divisibility criterion by
finite automata deserve special interest.

Lecomte and Rigo \cite{LR} showed the following: given a
regular language $L=\{w_0<w_1<\cdots\}$ genealogically ordered, extracting from $L$ words whose indices belong to an
ultimately periodic set $I\subset\mathbb{N}$ is a
regularity-preserving operation defining a language $L_I$. Krieger {\em et al.} \cite{Ra} considered the state complexity of this operation. 
If the minimal automaton of $L$ has $n$ states, it is natural to give
bounds or try to estimate the number of states of the minimal
automaton of $L_I$ as a function of $n$, the preperiod and period of $I$. Such results could be useful in solving the decidability question mentioned in the last paragraph. For example, Alexeev \cite{Al} recently gave an exact formula for the number of states of the
minimal automaton of the language $0^*\rep_b(m\mathbb{N})$, that is
the set of $b$-ary representations of the multiples of $m\ge 1$.



In this paper, we study the
state complexity for the divisibility criterion by $m\ge 2$ in the
framework of linear numeration systems. Let $0^*\rep_U(m\mathbb{N})$ be the language of greedy representations of the multiples of $m\ge 1$ in the numeration system $U$. Under some mild
assumptions, Theorem~\ref{the:main} gives the number of states of the
trim minimal automaton of $0^*\rep_U(m\mathbb{N})$ from which infinitely many words are accepted. As a corollary, we show that, for a certain class of numeration systems, we can give the precise number of states of this automaton. For
instance, for the Fibonacci numeration system, the corresponding
number of states is $2m^2$, see Corollary~\ref{cor:bonacci}. Finally we are able to give a lower bound for the state complexity of $0^*\rep_U(m\N)$ for any numeration system.

Note that the study of state complexity could possibly be related to the length
of the formulas describing such sets in a given numeration system. It is noteworthy that for linear
numeration systems whose characteristic polynomial is the minimal
polynomial of a Pisot number, $U$-recognizable sets can be
characterized by first order formulas of a convenient extension of
$\langle\mathbb{N},+\rangle$, see \cite{BH}.

Our result can only be fully understood when
one has a clear picture of $\mathcal{A}_U$, the trim minimal automaton recognizing the set $0^*\rep_U(\N)$ of all greedy representations. Such a description for a
linear numeration system satisfying the dominant root condition
(see below) is partially recalled in Theorem~\ref{the:descriptionNU} \cite{CRRW}.

\section{Background on Numeration Systems}\label{sec:background}
In this paper,  when we write $x = x_{n-1} \cdots x_{0}$ where $x$ is a word, we mean that $x_i$ is a letter for all $i \in \{0,\ldots,n-1\}$.

An increasing sequence $U=(U_n)_{n\ge 0}$ of integers is a
{\em numeration system}, or a {\em numeration basis}, if $U_0=1$ and $C_U:=\sup_{n\ge 0}\lceil\frac{U_{n+1}}{U_n}\rceil<+\infty$. We let $A_U$ be the alphabet $\{0,\ldots,C_U-1\}$. A greedy representation of a non-negative integer $n$ is a word $w=w_{\ell-1}\cdots w_0$ over $A_U$ satisfying
\[\sum_{i=0}^{\ell-1} w_iU_i=n \text{ and }\forall j\in\{1,\ldots,\ell\},\quad \sum_{i=0}^{j-1}w_iU_i<U_j.\]
We denote by $\rep_U(n)$ the greedy representation of $n>0$ satisfying
$w_{\ell-1}\neq 0$.
By convention, $\rep_U(0)$ is the empty word $\varepsilon$. 
The language $\rep_U(\N)$ is called the {\em numeration language}. A set
$X$ of integers is {\em $U$-recognizable} if $\rep_U(X)$ is regular,
i.e., accepted by a finite automaton.
If $\N$ is $U$-recognizable, then we let $\mathcal{A}_U=(Q_U,q_{U,0},F_U,A_U,\delta_U)$ denote the trim minimal
automaton of the language $0^*\rep_U(\mathbb{N})$ having
$\#\mathcal{A}_U$ states. The {\em numerical value map} $\val_U:A_U^*\to\mathbb{N}$ maps any
word $d_{\ell-1}\cdots d_0$ over $A_U$ to $\sum_{i=0}^{\ell-1} d_iU_i$. 

 \begin{definition}
     A numeration system $U=(U_n)_{n\ge 0}$ is said to be {\em linear}, 
     if there exist $k\ge1$ and $a_0,\ldots,a_{k-1}\in\mathbb{Z}$ such that
 \begin{equation}
     \label{eq:rec}
     \forall n\in\N,\ U_{n+k}=a_{k-1} U_{n+k-1}+\cdots +a_0 U_n.
 \end{equation}
We say that $k$ is the {\em length} of the recurrence relation.
 \end{definition}

\begin{theorem}[\cite{CANT}]\label{the:folk}
    Let $p,r\ge 0$. If $U=(U_n)_{n\ge 0}$ is a linear numeration
    system, then
$$\val_U^{-1}(p\,\mathbb{N}+r)=\{w\in A_U^*\mid 
\val_U(w)\in p\,\mathbb{N}+r \}$$
is accepted by a DFA that can be effectively constructed. In
particular, if $\mathbb{N}$ is $U$-recognizable, then any eventually periodic
set is $U$-recognizable.
\end{theorem}

%

\begin{definition}
If $U=(U_n)_{n\ge0}$ is a linear numeration system satisfying \[\lim_{n\to+\infty}\frac{U_{n+1}}{U_n}=\beta\] for some real $\beta> 1$, then it is said to {\em satisfy the dominant root condition} and $\beta$ is called the
{\em dominant root} of the recurrence.
\end{definition}

\begin{example}[(Fibonacci numeration system)]\label{exa:1}
    With $U_{n+2}=U_{n+1}+U_n$ and $U_0=1$, $U_1=2$, we get the usual
    Fibonacci numeration system. The Golden Ratio $(1+\sqrt{5})/2$ is the
    dominant root.  For
    this system, $A_U=\{0,1\}$ and $\mathcal{A}_U$ accepts all words
    over $A_U$ except those containing the factor $11$. 
\begin{figure}[htbp]
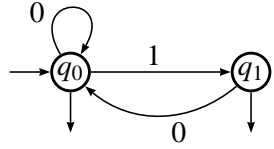

        \centering
\VCDraw{%
        \begin{VCPicture}{(-1,-1.7)(5,1.2)}
 \State[q_0]{(0,0)}{1}
 \State[q_1]{(4,0)}{2} 
\Initial{1}
\Final[s]{1}
\Final[s]{2}
\LoopN{1}{0}
\EdgeL{1}{2}{1}
\VArcL{arcangle=40}{2}{1}{0} 
\end{VCPicture}
}
        \caption{The automaton $\mathcal{A}_U$ for the Fibonacci numeration system.}
        \label{fig:fibonacci}
    \end{figure}
\end{example}

\begin{example}[($\ell$-bonacci numeration system)]\label{exa:nbonacci}
    Let $\ell\ge 2$. Consider the linear recurrence sequence defined by 
$$\forall n\in\N,\ U_{n+\ell}=\sum_{i=0}^{\ell-1}U_{n+i}$$
and for $i\in\{0,\ldots,\ell-1\}$, $U_i=2^i$.  For
    this system, $A_U=\{0,1\}$ and $\mathcal{A}_U$ accepts all words
    over $A_U$ except those containing the factor $1^\ell$. 
\begin{figure}[htbp]
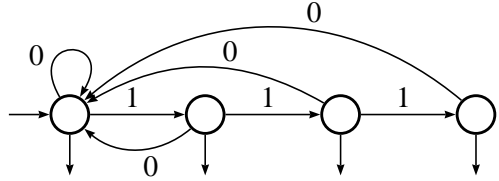

        \centering
\VCDraw{%
        \begin{VCPicture}{(-1,-1.7)(10,2)}
 \State{(0,0)}{1}
 \State{(3,0)}{2} 
 \State{(6,0)}{3} 
 \State{(9,0)}{4} 
\Initial{1}
\Final[s]{1}
\Final[s]{2}
\Final[s]{3}
\Final[s]{4}
\LoopN{1}{0}
\EdgeL{1}{2}{1}
\EdgeL{2}{3}{1}
\EdgeL{3}{4}{1}
\VArcL{arcangle=40}{2}{1}{0} 
\VArcR{arcangle=-30}{3}{1}{0} 
\VArcR{arcangle=-40}{4}{1}{0} 

\end{VCPicture}
}
        \caption{The automaton $\mathcal{A}_U$ for the $4$-bonacci numeration system.}
        \label{fig:lbonacci}
    \end{figure}
\end{example}

\begin{example}\label{exa:laurent}
     With
     $U_{n+2}=2U_{n+1}+U_n$, $U_0=1$, $U_1=3$, we have a
     numeration system 
 $$(U_n)_{n\ge 0}=1,3,7,17,41,99,239,\ldots$$ having a dominant root $\beta=1+\sqrt{2}$ and 
 where the corresponding automaton $\mathcal{A}_U$ is depicted in Figure~\ref{fig:sqrt2}.
 \begin{figure}[htbp]
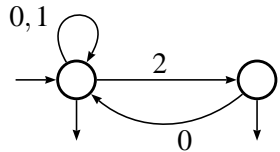

         \centering
 \VCDraw{%
         \begin{VCPicture}{(-1,-1.7)(5,1.2)}
  \State{(0,0)}{1}
  \State{(4,0)}{2} 
 \Initial{1}
 \Final[s]{1}
 \Final[s]{2}
 \LoopN{1}{0,1}
 \EdgeL{1}{2}{2}
 \VArcL{arcangle=40}{2}{1}{0} 
 \end{VCPicture}
 }
         \caption{The automaton $\mathcal{A}_U$ for the system having dominant root $1+\sqrt{2}$.}
         \label{fig:sqrt2}
    \end{figure}
\end{example}

Recall that the states of the minimal automaton of an arbitrary
language $L$ over an alphabet $A$ are given by the 
equivalence classes of the Myhill-Nerode congruence $\sim_L$, which is
defined by
\[\forall w,z\in A^*,\ w\sim_L z\Leftrightarrow \{x\in A^*\mid wx\in L\}=\{x\in A^*\mid zx\in L\}.\]
Equivalently, the states of the minimal automaton of $L$ correspond
to the sets $w^{-1}L=\{x\in A^*\mid wx\in L\}$. In this paper the symbol $\sim$ will be used
to denote Myhill-Nerode congruences.

\begin{definition}
A directed multi-graph is \emph{strongly connected} if for all pairs of
vertices $(s,t)$, there is a directed path from $s$ to $t$.  A {\em strongly
 connected component} of a directed multi-graph is a maximal strongly connected
subgraph.  Such a component is said to be {\em non-trivial} if
it does not consist of a single vertex with no loop. 
\end{definition}

\begin{theorem}\label{the:descriptionNU}\cite{CRRW}
    Let $U$ be a linear numeration system such that $\rep_U(\mathbb{N})$ is regular.
    \begin{itemize}
      \item[\textup{(i)}] The automaton $\mathcal{A}_U$ has a non-trivial strongly connected component
        $\mathcal{C}_U$ containing the initial state.
      \item[\textup{(ii)}] If $p$ is a state in $\mathcal{C}_U$, then there exists
        $N\in\N$ such that $\delta_U(p,0^n)=q_{U,0}$ for all $n\ge N$.
        In particular, if $q$ (resp.
        $r$) is a state in $\mathcal{C}_U$ (resp. not in
        $\mathcal{C}_U$) and if $\delta_U(q,\sigma)=r$, then
        $\sigma\neq 0$.
\item[\textup{(iii)}] If $\mathcal{C}_U$ is the only non-trivial strongly
  connected component of $\mathcal{A}_U$, then we have
 $\displaystyle\lim_{n\to+\infty}U_{n+1}-U_n=+\infty$.
        \item[\textup{(iv)}] If $\displaystyle\lim_{n\to+\infty}U_{n+1}-U_n=+\infty$,
then the state $\delta_U(q_{U,0},1)$ belongs to $\mathcal{C}_U$.
 \end{itemize}

\noindent In the case where the numeration system $U$ has a dominant root
$\beta > 1$, if the automaton $\mathcal{A}_U$ has more than one non-trivial
strongly connected component, then any such component distinct from
$\mathcal{C}_U$ is restricted to a cycle all of whose edges are labelled $0$.
\end{theorem}

\section{State complexity for divisibility criterion}

\begin{definition}\label{def:k}
    Let $U=(U_n)_{n\ge 0}$ be a numeration system and $m\ge 2$ be an
    integer. The sequence $(U_n\bmod{m})_{n\ge 0}$ satisfies a linear
    recurrence relation of minimal length. This integer is denoted by
    $k_{U,m}$ or simply by $k$ if the context is clear. This quantity
    is given by the largest $t$ such that
$$\det H_t\not\equiv 0 \pmod{m}, \text{ where }
H_t=\begin{pmatrix}
    U_0&U_1&\cdots &U_{t-1}\\
    U_1&U_2&\cdots &U_t\\
\vdots&\vdots&\ddots &\vdots\\
U_{t-1}&U_t&\cdots & U_{2t-2}
\end{pmatrix}.$$
\end{definition}

\begin{example}\label{exa:laurent2}
    Let $m=2$ and consider the sequence introduced in
    Example~\ref{exa:laurent}. The sequence $(U_n\bmod{2})_{n\ge 0}$
    is constant and trivially satisfies the recurrence relation $U_{n+1}=U_n$ with
    $U_0=1$. Therefore, we get $k_{U,2}=1$. For $m=4$, one can check that
    $k_{U,4}=2$.
\end{example}

\begin{definition}
    Let $U=(U_n)_{n\ge 0}$ be a numeration system and $m\ge 2$ be an
    integer. Let $k=k_{U,m}$. Consider the system of linear equations
    $$H_k\, \mathbf{x}\equiv\mathbf{b}\pmod m$$ where $H_k$ is the $k\times k$ matrix given in
    Definition~\ref{def:k}.  We let $S_{U,m}$ denote the number of
    $k$-tuples $\mathbf{b}$ in $\{0,\ldots,m-1\}^k$ such that the system
    $H_k\, \mathbf{x}\equiv\mathbf{b}\bmod{m}$ has at least one solution.
\end{definition}

\begin{example}
    Again take the same recurrence relation as in Example~\ref{exa:laurent} and $m=4$.
    Consider the system 
$$\left\{
\begin{array}{rcl}
1\, x_1+3\, x_2&\equiv&b_1\pmod{4}\\
3\, x_1+7\, x_2&\equiv&b_2\pmod{4}\\
\end{array}\right.$$
We have $2x_1\equiv b_2-b_1\pmod{4}$. Hence for each value of $b_1$ in
$\{0,\ldots,3\}$, $b_2$ can take at most $2$ values. One can therefore
check that $S_{U,4}=8$.
\end{example}

\begin{remark}
  Let $\ell\ge
  k=k_{U,m}$. 
  Then the number of $\ell$-tuples $\mathbf{b}$ in
  $\{0,\ldots,m-1\}^\ell$ such that the system $H_\ell\,
  \mathbf{x}\equiv\mathbf{b}\pmod{m}$ has at least one solution equals
  $S_{U,m}$. Let us show this assertion for $\ell=k+1$.  Let $H_\ell'$
  denote the $\ell\times k$ matrix obtained by deleting the last
  column of $H_\ell$ and let ${\bf x}'$ denote the $k$-tuple obtained
  by deleting the last element of ${\bf x}$.  Observe that the
  $\ell$-th column of $H_\ell$ is a linear combination of the other
  columns of $H_\ell$.  It follows that if ${\bf b} =
  (b_0,\ldots,b_{k-1},b) \in\{0,\ldots,m-1\}^\ell$ is an $\ell$-tuple
  for which the system $H_\ell'{\bf x}' \equiv {\bf b} \pmod m$ has a
  solution, then ${\bf b}' = (b_0,\ldots,b_{k-1})
  \in\{0,\ldots,m-1\}^k$ is a $k$-tuple for which the system $H_k{\bf
    x}' \equiv {\bf b}' \pmod m$ also has a solution.  Furthermore,
  the $\ell$-th row of $H_\ell'$ is a linear combination of the other
  rows of $H_\ell'$, so for every such ${\bf b}'$, there is exactly
  one ${\bf b}$ such that $H_\ell'{\bf x}' \equiv {\bf b} \pmod m$ has
  a solution.  This establishes the claim.
\end{remark}

We define two properties that $\mathcal{A}_U$ may satisfy:
\begin{enumerate}
  \item[\textup{(H.1)}] $\mathcal{A}_U$ has a single strongly connected component denoted by $\mathcal{C}_U$,
  \item[\textup{(H.2)}] for all states $p,q$ in $\mathcal{C}_U$, with
    $p\neq q$, there exists a word $x_{pq}$ such that
    $\delta_U(p,x_{pq})\in\mathcal{C}_U$ and $\delta_U(q,x_{pq})\not\in\mathcal{C}_U$, or,
    $\delta_U(p,x_{pq})\not\in\mathcal{C}_U$ and $\delta_U(q,x_{pq})\in\mathcal{C}_U$.
\end{enumerate}

\begin{theorem}\label{the:main}
    Let $m\ge 2$ be an integer. Let $U=(U_n)_{n\ge 0}$ be a linear 
    numeration system satisfying the recurrence relation \eqref{eq:rec} 
	such that
    \begin{enumerate}
      \item[\textup{(a)}] $\mathbb{N}$ is $U$-recognizable and $\mathcal{A}_U$ satisfies the assumptions \textup{(H.1)} and \textup{(H.2)},
      \item[\textup{(b)}] $(U_n\bmod m)_{n\ge 0}$ is purely periodic. 
    \end{enumerate}

Then the number of states of the trim minimal automaton $\mathcal{A}_{U,m}$ of the language $$0^*\rep_U(m\mathbb{N})$$ from which infinitely many words are accepted is
$$(\#\mathcal{C}_U) S_{U,m}.$$
\end{theorem}

From now on we fix an integer $m\ge 2$ and a numeration system $U=(U_n)_{n\ge 0}$ satisfying the recurrence relation \eqref{eq:rec} 
and such that $\N$ is $U$-recognizable.

\begin{definition}
     We define a relation $\equiv_{U,m}$ over $A_U^*$.
    For all $u,v\in A_U^*$, $$u\equiv_{U,m} v\Leftrightarrow\left\{
    \begin{array}{l}
        u\sim_{0^*\rep_U(\mathbb{N})}v \quad \text{ and}\\
\forall i\in\{0,\ldots,k_{U,m}-1\},\ \val_U(u0^i)\equiv\val_U(v0^i)\pmod{m}\\
    \end{array}\right.$$
where $\sim_{0^*\rep_U(\mathbb{N})}$ is the Myhill-Nerode equivalence
for the language $0^*\rep_U(\mathbb{N})$ accepted by $\mathcal{A}_U$.
\end{definition}

\begin{lemma}\label{lem:cn}
    Let $u,v,x\in A_U^*$. If $u\equiv_{U,m} v$ and
    $ux,vx\in\rep_U(\mathbb{N})$, then $ux\equiv_{U,m}vx$ and in
    particular, $\val_U(ux)\equiv\val_U(vx)\pmod{m}$.
\end{lemma}

\begin{proof}
 By assumption, for all
$i\in\{0,\ldots,k-1\}$, $\val_U(u0^i)\equiv\val_U(v0^i)\pmod{m}$. Hence, for all
$i\in\{0,\ldots,k-1\}$, $a_i\val_U(u0^i)\equiv a_i\val_U(v0^i)\pmod{m}$.
Assume that $u=u_{\ell-1}\cdots u_0$. Note that
$$\sum_{i=0}^{k-1}a_i\val_U(u0^i)=\sum_{j=0}^{\ell-1} u_j \sum_{i=0}^{k-1} a_i U_{j+i}=
\sum_{j=0}^{\ell-1} u_j U_{j+k}=\val_U(u0^k).$$ Therefore, we can conclude
that $\val_U(u0^k)\equiv\val_U(v0^k)\pmod{m}$. Iterating this
argument, we have, for all $n\ge 0$,
\begin{equation}
    \label{eq:mod}
    \val_U(u0^n)\equiv\val_U(v0^n)\pmod{m}.
\end{equation}
Since the Myhill-Nerode
relation is a right congruence, we have that
$$ux\sim_{0^*\rep_U(\mathbb{N})}vx.$$ Let $i\in\{0,\ldots,k-1\}$. From
\eqref{eq:mod}, we deduce that
$$\val_U(u0^{|x|+i})+\val_U(x0^i)
\equiv\val_U(v0^{|x|+i})+\val_U(x0^i)\pmod{m}$$
and therefore $\val_U(ux0^i)\equiv\val_U(vx0^i)\pmod{m}$.
\end{proof}

\begin{proposition}\label{pro:main}
    Assume that the numeration system $U$ satisfies the
    assumptions of Theorem~\ref{the:main}.  Let $u,v\in A_U^*$ be such
    that $\delta_U(q_{U,0},u)$ and $\delta_U(q_{U,0},v)$ belong to $\mathcal{C}_U$. We
    have $u\equiv_{U,m} v$ if and only if
    $u\sim_{0^*\rep_U(m\mathbb{N})} v$.
\end{proposition}

\begin{proof}
From (b) the sequence $(U_n\bmod m)_{n\ge 0}$ is purely periodic, say of
    period $p$.  

%

Assume that $u\not\equiv_{U,m} v$. Our aim is to show that there
exists a word $y\in A_U^*$ that distinguishes $u$ and $v$ in the minimal automaton of $0^*\rep_U(m\mathbb{N})$, i.e., either $uy\in 0^*\rep_U(m\mathbb{N})$ and $vy\not\in 0^*\rep_U(m\mathbb{N})$,
or $uy\not\in 0^*\rep_U(m\mathbb{N})$ and $vy\in 0^*\rep_U(m\mathbb{N})$. 

As a first case, assume $u\not\sim_{0^*\rep_U(\mathbb{N})}v$. Since
$\delta_U(q_{U,0},u)$ and $\delta_U(q_{U,0},v)$ both belong to $\mathcal{C}_U$, this means
that $\delta_U(q_{U,0},u)$ and $\delta_U(q_{U,0},v)$ are two different states in
$\mathcal{C}_U$. By (H.2), without loss of generality, we may assume that there exists a word $x$
such that
$$\delta_U(q_{U,0},ux)\in\mathcal{C}_U \quad\text{ and }\quad \delta_U(q_{U,0},vx)\not\in\mathcal{C}_U.$$
Since $\mathcal{A}_U$ contains only one strongly connected component, only finitely many words may be accepted from $\delta_U(q_{U,0},vx)$. Let $T$ be the length of the
longest word accepted from $\delta_U(q_{U,0},vx)$. Let $i\in\{1,\ldots,m\}$ be such that
$\val_U(ux)+i\equiv 0\pmod{m}$. Using properties \textup{(ii)--(iv)} from Theorem~\ref{the:descriptionNU} $i$ times and the fact that $\delta_U(q_{U,0},1)$ is finite, there exist
$r_1,\ldots,r_i\ge 0$ such that the word
$$y=x(0^{r_1p}0^{p-1}1)(0^{r_2p}0^{p-1}1)\cdots (0^{r_ip}0^{p-1}1)$$
has a length larger than $T+|x|$ and is such that $uy$ is a greedy representation. Moreover, due to the periodicity of $(U_n\bmod m)_{n\ge 0}$, we have $\val_U(uy)\equiv 0 \pmod{m}$ and therefore $uy$ belongs to $0^*\rep_U(m\mathbb{N})$. Hence, the word $y$ distinguishes $u$ and $v$ for the language $0^*\rep_U(m\mathbb{N})$.

Now assume that $u\sim_{0^*\rep_U(\mathbb{N})}v$ and there exists
$j\in\{0,\ldots,k-1\}$ such that
$\val_U(u0^j)\not\equiv\val_U(v0^j)\pmod{m}$. There exists $i<m$ such
that $\val_U(u0^j)+i\equiv 0\pmod{m}$ and $\val_U(v0^j)+i\not\equiv
0\pmod{m}$. Using properties \textup{(ii)-(iv)} from Theorem~\ref{the:descriptionNU} there exist
$s_1,\ldots,s_i\ge 0$ such that the word
$$y=(0^{s_1p}0^{p-1}1)(0^{s_2p}0^{p-1}1)\cdots (0^{s_ip}0^{p-1}1)$$
distinguishes $u$ and $v$.

Consider the other implication and assume that $u\equiv_{U,m} v$. Let
$x$ be a word such that $ux\in 0^*\rep_U(m\mathbb{N})$. From Lemma~\ref{lem:cn}, we only have to show that
$vx$ is a greedy representation.  Since $v$ is a 
greedy representation and
$u\sim_{0^*\rep_U(\mathbb{N})}v$, we can conclude that $vx$ is a
greedy representation. Hence the conclusion follows.
\end{proof}

\begin{proof}[Proof of Theorem~\ref{the:main}] 
    If $u$ is a word such that $\delta_U(q_{U,0},u)$ belongs to
    $\mathcal{C}_U$, then with the same reasoning as in the proof of
    Proposition~\ref{pro:main}, there exist infinitely many words $x$
    such that $ux\in 0^*\rep_U(m\mathbb{N})$. On the other hand, by (H.1), if
    $v$ is a word such that $\delta_U(q_{U,0},v)$ does not belong to
    $\mathcal{C}_U$, there exist finitely
    many words $x$ such that $vx\in 0^*\rep_U(m\mathbb{N})$. Therefore, the number of states of the trim minimal automaton of the language $0^*\rep_U(m\mathbb{N})$ from which infinitely many words are accepted is the number of sets $u^{-1}0^*\rep_U(m\N)$ where $u$ is a word over $A_U$ such that 
$\delta_U(q_{U,0},u)$ belongs to $\mathcal{C}_U$.
Hence, as a
    consequence of Proposition~\ref{pro:main}, this number is also the number of equivalence classes
    $[u]_{\equiv_{U,m}}$ with $u$ being such that $\delta_U(q_{U,0},u)\in\mathcal{C}_U$.  What we  have to do to conclude the proof
    is therefore to count the number of such equivalence classes.

First we show that there are at most $\#\mathcal{C}_U S_{U,m}$ such classes. By definition, if $u,v\in A_U^*$ are such that $\delta_U(q_{U,0},u)\neq\delta_U(q_{U,0},v)$, then $u\not\equiv_{U,m}v$. Otherwise, $u\not\equiv_{U,m}v$ if and only if there exists $\ell<k$ such that $\val_U(u0^\ell)\not\equiv\val_U(v0^\ell)\pmod{m}$. 

Let $u=u_{r-1}\cdots u_0\in A_U^*$. 
We let $\mathbf{b}_u$ denote the $k$-tuple $(b_0,\ldots,b_{k-1})^T\in\{0,\ldots,m-1\}^k$ defined by
    \begin{equation}
        \label{eq:bu}
\forall s\in\{0,\ldots,k-1\},\ \val_U(u0^s)\equiv b_s\pmod{m}.
    \end{equation}
Using the fact that the sequence $(U_n)_{n\ge 0}$ satisfies \eqref{eq:rec}, there exist $\alpha_0,\ldots,\alpha_{k-1}$ such that 
\begin{equation}
\label{eq:bu1}\forall s\in\{0,\ldots,k-1\},\ 
\val_U(u0^s)=\sum_{i=0}^{r-1}u_iU_{i+s}=\sum_{i=0}^{k-1}\alpha_{i} U_{i+s}.
\end{equation}
Using \eqref{eq:bu} and \eqref{eq:bu1}, 
we see that the system $H_k \mathbf{x}\equiv\mathbf{b}_u\pmod{m}$ has a solution $\mathbf{x}=(\alpha_0,\ldots,\alpha_{k-1})^T$. 

If $u,v\in A_U^*$ are such that $\delta_U(q_{U,0},u)=\delta_U(q_{U,0},v)$ but $u\not\equiv_{U,m}v$, then $\mathbf{b}_u\neq\mathbf{b}_v$. From the previous paragraph the systems $H_k \mathbf{x}\equiv\mathbf{b}_u\pmod{m}$ and $H_k \mathbf{x}\equiv\mathbf{b}_v\pmod{m}$ both have a solution. Therefore, there are at most $\#\mathcal{C}_U S_{U,m}$ infinite equivalence classes.

Second we show that there are at least $\#\mathcal{C}_U S_{U,m}$ such classes. Let
$\mathbf{c}=(c_0,\ldots,c_{k-1})^T\in\{0,\ldots,m-1\}^k$ be such that the system $H_k \mathbf{x}\equiv\mathbf{c}\pmod{m}$ has a solution $\mathbf{x_c}=(\alpha_0,\ldots,\alpha_{k-1})^T$. 
Let $q$ be any state in $\mathcal{C}_U$. Our aim is to build a word $y$ over $A_U$ such that $$\delta_U(q_{U,0},y)=q \text{ and }\forall s\in\{0,\ldots,k-1\},\ \val_U(y0^s)\equiv c_s\pmod{m}.$$

Since $\mathcal{A}_U$ is
accessible, there exists a word $u\in A_U^*$ such that $\delta_U(q_{U,0},u)=q$. With this word $u$ is associated a unique $\mathbf{b}_u=(b_0,\ldots,b_{k-1})^T\in\{0,\ldots,m-1\}^k$ given by \eqref{eq:bu}. The system $H_k \mathbf{x}\equiv \mathbf{b}_u\pmod{m}$ has a solution denoted by $\mathbf{x}_u$.

Define $\gamma_0,\ldots,\gamma_{k-1}\in\{0,\ldots,m-1\}$ by  $\mathbf{x_c}-\mathbf{x}_u\equiv(\gamma_0,\ldots,\gamma_{k-1})^T\pmod{m}$. Thus
\begin{equation}\label{eq:gamma}
 H_k (\mathbf{x}_c-\mathbf{x}_u)\equiv\mathbf{c}-\mathbf{b_u}\pmod{m}.
\end{equation}
 
Using properties \textup{(ii)--(iv)} from Theorem~\ref{the:descriptionNU} from the initial state $q_{U,0}$, there exist $t_{1,1},\ldots,t_{1,\gamma_0}$ such that the word 
$$w_1=(0^{pt_{1,1}} 0^{p-1} 1)\cdots (0^{pt_{1,\gamma_0}} 0^{p-1} 1)$$
satisfies $\delta_U(q_{U,0},w_1)\in\mathcal{C}_U\cap F_U$ and $\val_U(w_1)\equiv\gamma_0 U_0\pmod{m}$. We can iterate this construction. For $j\in\{2,\ldots,k\}$, there exist $t_{j,1},\ldots,t_{j,\gamma_j}$ such that the word 
$$w_j=w_{j-1}(0^{pt_{j,1}} 0^{p-j} 1 0^{j-1})\cdots (0^{pt_{j,\gamma_j}} 0^{p-j} 1 0^{j-1})$$
satisfies $\delta_U(q_{U,0},w_j)\in\mathcal{C}_U\cap F_U$ and $\val_U(w_j)\equiv\val_U(w_{j-1})+\gamma_{j-1} U_{j-1}\pmod{m}$. Consequently, we have  
$$\val_U(w_k)\equiv\gamma_{k-1}U_{k-1}+\cdots+\gamma_0U_0\pmod{m}.$$
Now take $r$ and $r'$ large enough such that $\delta_U(q_{U,0},w_{k}0^{rp})=q_{U,0}$ and $r'p\ge|u|$. Such an $r$ exists by \textup{(ii)} in Theorem~\ref{the:descriptionNU}. The word
$$y=w_k 0^{(r+r')p-|u|}u$$
is such that $\delta_U(q_{U,0},y)=\delta_U(q_{U,0},u)=q$ and taking into account the periodicity of $(U_n\bmod{m})_{n\ge 0}$, we get 
$$\val_U(y)\equiv\val_U(w_k)+\val_U(u)\pmod{m}.$$
In view of \eqref{eq:gamma}, we obtain
$$\forall s\in\{0,\ldots,k-1\},\ \val_U(y0^s)\equiv\sum_{i=0}^{k-1}\gamma_iU_{i+s}+b_s \equiv c_s-b_s+b_s=c_s
\pmod{m}.$$
\end{proof}

\begin{corollary}\label{cor:strong}
Assume that the numeration system $U$ satisfies the assumptions of Theorem~\ref{the:main}. Assume moreover that $\mathcal{A}_U$ is strongly connected (i.e. $\mathcal{A}_U=\mathcal{C}_U$). Then the number of states of the trim minimal automaton of the language $0^*\rep_U(m\mathbb{N})$ is
$(\#\mathcal{C}_U) S_{U,m}$.
\end{corollary}

\begin{proof}
We use the same argument as in the beginning of the proof of Theorem~\ref{the:main}. Since $\mathcal{A}_U=\mathcal{C}_U$, all of the sets $u^{-1}0^*\rep_U(m\N)$ are infinite. Hence, infinitely many words are accepted from any state of $\mathcal{A}_{U,m}$.
\end{proof}


\begin{corollary}\label{cor:bonacci}
    Let $\ell\ge 2$. For the $\ell$-bonacci numeration system
    $U=(U_n)_{n\ge 0}$ defined by $U_{n+\ell}=U_{n+\ell-1}+\cdots+U_n$
    and $U_i=2^i$ for all $i<\ell$, the number of states of the trim
    minimal automaton of the language $0^*\rep_U(m\mathbb{N})$ is $\ell.m^\ell$.
\end{corollary}

 \begin{proof}
     First note that the trim minimal automaton of
     $0^*\rep_U(\mathbb{N})$ consists of a unique strongly connected
     component made of $\ell$ states (see Figure~\ref{fig:fibonacci}) and $\mathcal{A}_U$ satisfies all
     the required assumptions. 
The matrix
     $\mathbf{H}_\ell$ has a determinant equal to $\pm 1$.
     Therefore, for all $\mathbf{b}\in\{0,\ldots,m-1\}^\ell$, the
     system $\mathbf{H}_\ell\mathbf{x}\equiv\mathbf{b}\pmod{m}$ has a solution.
     There are $m^\ell$ such vectors $\mathbf{b}$. We conclude by using Corollary~\ref{cor:strong}.
 \end{proof}

 To build the minimal automaton of $\rep_U(m\mathbb{N})$, one can use
 Theorem~\ref{the:folk} to first have an automaton accepting the
 reversal of the words over $A_U$ whose numerical value is divisible by
 $m$. We consider the reversal representation, that is least significant
 digit first, to be able to handle the period\footnote{Another option 
   is to consider a non-deterministic finite automaton reading most
   significant digits first.} of $(U_n\bmod{m})_{n\ge 0}$. Such an
 automaton has $m$ times the length of the period of
 $(U_n\bmod{m})_{n\ge 0}$ states. Then minimizing the intersection of
 the reversal of this automaton with the automaton $\mathcal{A}_U$, we
 get the expected minimal automaton of $0^*\rep_U(m\mathbb{N})$.

 Taking advantage of Proposition~\ref{pro:main}, we get an automatic
 procedure to obtain directly the minimal automaton $\mathcal{A}_{U,m}$
 of $0^*\rep_U(m\mathbb{N})$. States of $\mathcal{A}_{U,m}$ are given
 by $(k+1)$-tuples. The state reached by reading $w$ has as first
 component the state of $\mathcal{A}_U$ reached when reading $w$ and
 the other components are
 $\val_U(w)\bmod{m},\ldots,\val_U(w0^{k-1})\bmod{m}$.

 \begin{example}
     Consider the Fibonacci numeration system and $m=3$. The states of
     $\mathcal{A}_U$ depicted in Figure~\ref{fig:fibonacci} are denoted
     by $q_0$ and $q_1$. The states of $\mathcal{A}_{U,3}$ are
     $r_0,\ldots,r_{17}$. The transition function of $\mathcal{A}_{U,3}$ is denoted by $\tau$.
 $$\begin{array}{c|c|c|c}
 w & r=(\delta_U(q_0,w),\val_U(w),\val_U(w0)) & \tau(r,0) & \tau(r,1)\\
 \hline 
 \varepsilon, 0, 10^310 & r_0=(q_0,0,0) & r_0 & r_1 \\
 1 &  r_1=(q_1,1,2) & r_2 &  \\
 10,10100 & r_2=(q_0,2,0) & r_3 & r_4 \\
 100 & r_3=(q_0,0,2) & r_5& r_6\\
 101 & r_4=(q_1,1,1) & r_7 & \\
 1000,(10)^3 & r_5=(q_0,2,2) & r_8 & r_9 \\
 1001 & r_6=(q_1,0,1) & r_{10} & \\
 1010, (100)^2 & r_7=(q_0,1,2) & r_2 & r_{11}\\
 10^4,10^410 & r_8=(q_0,2,1)     & r_{12} & r_{13} \\
 10^31 & r_9=(q_1,0,0)     & r_0 & \\
 10010,10^7 & r_{10}=(q_0,1,1)  & r_7 & r_{14} \\
 10101 & r_{11}=(q_1,0,2)  & r_5 & \\
 10^5 & r_{12}=(q_0,1,0)   & r_{15} & r_{16} \\
 10^41 & r_{13}=(q_1,2,2)  & r_8 & \\
 100101 & r_{14}=(q_1,2,1) & r_{12} & \\
 10^6 & r_{15}=(q_0,0,1)   & r_{10} & r_{17} \\
 10^51 & r_{16}=(q_1,1,0)  & r_{15} & \\
 10^61 & r_{17}=(q_1,2,0) & r_3 & \\ 
 \end{array}$$
 \end{example}

\begin{definition}\label{def:bertrand}
    A numeration system $U=(U_n)_{n\ge 0}$ is a {\em Bertrand
      numeration system} if, for all $w\in A_U^+$,
    $w\in\rep_U(\mathbb{N})\Leftrightarrow w0\in\rep_U(\mathbb{N})$.
\end{definition}

All the systems presented in Examples~\ref{exa:1},~\ref{exa:nbonacci}
and \ref{exa:laurent} are Bertrand numeration systems. As a
consequence of Parry's Theorem \cite{parry,Lot} and Bertrand's theorem
\cite{Be,Lot}, the canonical automaton $\mathcal{A}_\beta$ associated
with $\beta$-expansions is a trim minimal automaton (therefore, any
two distinct states are distinguished) which is moreover strongly
connected. The following result is therefore obvious.

\begin{proposition}\label{pro:ber}
    Let $U$ be the Bertrand numeration system associated with a non-integer Parry
    number $\beta>1$. The set $\mathbb{N}$ is $U$-recognizable and the
    trim minimal automaton $\mathcal{A}_U$ of $0^*\rep_U(\mathbb{N})$
    fulfills properties \textup{(H.1)} and \textup{(H.2)}.
\end{proposition}

We can therefore apply Theorem~\ref{the:main} to this class of Bertrand
numeration systems.

Finally, we give a lower bound when the numeration system satisfies weaker hypotheses than those of Theorem~\ref{the:main}.

\begin{proposition}
Let $U$ be any numeration system (not necessarily linear). The number of state of $\mathcal{A}_{U,m}$ is at least $|\rep_U(m)|$.
\end{proposition}

\begin{proof}
Let $n=|\rep_U(m)|$. For each $i\in\{1,\ldots,n\}$, we define $p_i$ (resp. $s_i$) to be the prefix (resp. suffix) of length $i$ (resp. $n-i$) of  $\rep_U(m)$. We are going to prove that for all $i,j\in\{1,\ldots,n\}$, we have $p_i\not\sim_{0^*\rep_U(m\N)} p_j$.  Let $i,j\in\{1,\ldots,n\}$. We may assume that $i<j$. Obviously, the word $p_js_j$ belongs to $0^*\rep_U(m\N)$. On the other hand, observe that $|p_is_j|\in\{1,\ldots,n-1\}$. Therefore the word $p_is_j$ does not belong to $0^*\rep_U(m\N)$ since it cannot simultaneously be greedy and satisfy $\val_U(p_is_j)\equiv 0\pmod{m}$.
Hence, the word $s_j$ distinguishes $p_i$ and $p_j$.
\end{proof}

\section{Perspectives}
\begin{itemize}
\item With the same assumptions as in Theorem~\ref{the:main}, can we count the number of states from which only finitely many words are accepted?
\item Can we weaken the assumptions of Theorem~\ref{the:main}?
\item If $X$ is a finite union of arithmetic progressions, can we give bounds for the number of states of the trim minimal automaton accepting $0^*\rep_U(X)$?
\end{itemize}

\bibliographystyle{eptcs}

\end{document}